\documentclass{article}
\usepackage{amssymb,amsmath,amsthm,a4wide}
\advance\textheight by \topskip
\usepackage{tikz}

\DeclareMathAlphabet{\mathbfit}{OML}{cmm}{b}{it}
\def\I{\mathbfit I}
\def\cx{\mathbfit c}
\def\cto{\,---\,}
\def\phi{\varphi}
\newtheorem{theorem}{Theorem}
\newtheorem{lemma}[theorem]{Lemma}
\newtheorem{claim}[theorem]{Claim}
\newtheorem{fact}[theorem]{Fact}

\newtheorem{definition}{Definition}
\newcommand\eqdef{\stackrel{\mathrm{def}}{=}}
\def\|{{\,{|}\,}}
\makeatletter\def\@biblabel#1{#1.}\makeatother

\title{\bf Secret sharing on large girth graphs}
\author{L\'aszl\'o Csirmaz, P\'eter Ligeti
\thanks{Central European University, E\"otv\"os Lor\'and University and
R\'enyi Institute, Budapest}
\thanks{e-mail:~csirmaz@renyi.hu, turul@cs.elte.hu}}
\date{}

\begin{document}
\maketitle

\begin{abstract}
\noindent
We investigate graph based secret sharing schemes and its information ratio,
also called complexity, measuring the maximal amount of information the
vertices has to store.  It was conjectured that in large girth graphs, where
the interaction between far away nodes is restricted to a single path, this
ratio is bounded.  This conjecture was supported by several result, most
notably by a result of Csirmaz and Ligeti \cite{csl} saying that the
complexity of graphs with girth at least six and no neighboring high degree
vertices is strictly below 2.  In this paper we refute the above conjecture.
First, a family of $d$-regular graphs is defined iteratively such that the
complexity of these graphs is the largest possible $(d+1)/2$ allowed by
Stinson's bound \cite{stinson}.  This part extends earlier results of van
Dijk \cite{dijk} and Blundo et al \cite{bssv}, and uses the so-called
entropy method.  Second, using combinatorial arguments, we show that this
family contains graphs with arbitrary large girth.  In particular, we
obtain the following purely combinatorial result, which might be interesting
on its own: there are $d$-regular graphs with arbitrary large girth such
that any fractional edge-cover by stars (or by complete multipartite graphs)
must cover some vertex $(d+1)/2$ times.

\end{abstract}

\section{Introduction}\label{sec:intro}

\subsection{Motivation and notion}

{\em Secret sharing} is a method for distributing some secret information
between a set of participants by giving them partial knowledge of the secret
in a way that only pre-described coalitions will be able to reconstruct the
original secret from their respective parts.  More precisely, let $\cal{P}$
denote the set of {\em participants}.  A family of subsets $\cal{A}\subset$
$2^{\cal{P}}$ is called {\em access structure} if it is monotone increasing,
i.e.  if $A \in \cal{A}$ and $A\subseteq B$ then $B\in\cal{A}.$ The elements
of $\cal{A}$ are called \emph{qualified subsets}.

\begin{definition}\label{perfect_def} A  {\em perfect secret sharing}
$\cal{S}$ realizing the access structure $\cal{A}$ is a collection of random
variables $\xi_i$ for every $i\in \cal{P}$ and $\xi_s$ with a joint
distribution such that
\begin{itemize}
\item[(i)]  if $A\in \cal{A}$, then $\{\xi_i: i\in A\}$ determines $\xi_s;$
 \item[(ii)] if $A\notin \cal{A}$, then $\{\xi_i: i\in A\}$ is independent of
$\xi_s.$
\end{itemize}
\end{definition}

In this paper we focus on the special case when every minimal element of
$\cal{A}$ has two elements; these are the so-called {\em graph-based
schemes}.  The participants are the vertices of a graph $G=(V,E)$, and a set
of participants is qualified if there is an edge $e\in E$ with endpoints in
this set.

Traditionally the efficiency of a scheme is measured by the maximal amount
of information some participant has to store compared to the size of the
secret. This amount is called {\em information ratio}, or {\em complexity}
of the scheme. The {\em complexity} of a graph $G$ is the best lower bound
on the complexity of schemes realizing $G$, and can be defined formally in
following way:

\begin{definition}
The {\em complexity} (or information ratio) of the graph $G=(V,E)$ is
$$ \cx(G)=   \inf_{{\mathcal{S}}}\, \max_{v\in V} \,
\frac{\mathop{\hbox{\bf H}}(\xi_v)}{\mathop{\hbox{\bf H}}(\xi_s)},
$$
where the infimum is taken over all perfect secret sharing schemes
$\mathcal{S}$ realizing the access structure defined by the graph $G$,
and $\mathop{\hbox{\bf H}}(\xi)$ is the Shannon entropy of the random
variable $\xi$.
\end{definition}

One of the most interesting and challenging problems in this topic is to
determine the information ratio exactly for particular graphs or families of
graphs \cite{beimel}.  The main tool is to prove general or specific
estimations for the information ratio and find constructions which realize
these bounds.

\subsection{Related works}

For a comprehensible account on secret sharing, on its relevance, please
consult the survey \cite{beimel}.  In terms of the number of vertices, the
best lower bound for $\cx(G)$ established so far is logarithmic in the number
of vertices \cite{bssv,csirmcube,dijk}; while it is known that $\cx(G)$ is
always $\le c\cdot n/\log n$ for some small constant $c$ \cite{cslt,ep}.

The only known method for proving lower bounds is the entropy technique, we
explain it in more detail in Section \ref{sec:method}.  On the other hand,
every graph based secret sharing scheme yields an upper bound.  One
fundamental tool in constructing such schemes is Stinson's decomposition
technique.  While it is more general, we will use the following special case
only.

\begin{theorem}[Stinson's Decomposition Theorem, \cite{stinson}]
Suppose that the edges of the graph $G$ can be fractionally covered
by the edges of complete multipartite graphs such
that every vertex is covered with weight at most $k$. Then $\cx(G)\le k$.
\qed
\end{theorem}

Using Stinson's decomposition theorem for all star (spanned) subgraphs of a
graph $G$ with weight 0.5, it follows immediately that the complexity of $G$
is at most $(d+1)/2$ where $d$ is the maximal degree.  This upper bound is
known as the {\em Stinson's bound}.

To attack the problem whether Stinson's bound is tight, in \cite{dijk} van
Dijk defined a family of $d$-regular graphs and proved that Stinson's bound
is asymptotically tight for that family.  Later Blundo et al \cite{bssv}
showed that the complexity of this family actually matches Stinson's bound. 
In Section \ref{const_sec} we generalize their result by showing that
Stinson's bound is tight for an even larger family of $d$-regular graphs.

The graphs in van Dijk's family have girth 6 and ratio $(d+1)/2$.  The
complexity of other (sporadic) infinite graph families has been determined
as well.  Some examples are
\begin{itemize}
\item the edge graph of the $d$-dimensional hypercube which has
complexity $d/2$ and girth $4$ \cite{csirmcube};
\item
trees have complexity $<2$ and girth 0, for exact values consult \cite{cst};
\item
graphs with girth $>5$ and no adjacent vertices of degree at least three have
complexity strictly below $2$ \cite{csl}.
\end{itemize}
These results supported the intuition that high complexity requires high
connectivity and bounded girth, as the interaction between nodes diminishes
exponentially as their distance grows.  We show that this intuition was
wrong by constructing graphs with arbitrary large complexity {\em and}
arbitrary large girth at the same time.  The result is achieved in two
steps.  First, a family $\mathcal G_d$ of $d$-regular graphs is defined
which extends the above mentioned van Dijk's graph family.  Using the
entropy method, explained in Section \ref{sec:method}, we show that all
graphs in $\mathcal G_d$ have the maximal complexity allowed by Stinson's
bound, namely $(d+1)/2$.  Second, in Section \ref{graph_sec}, using purely
combinatorial arguments, we show that $\mathcal G_d$ contains graphs with
arbitrary large girth.

Any fractional cover of the edges by spanned stars automatically gives a
secret sharing scheme with complexity equal to the maximal cover weight on
the vertices.  Consequently any fractional edge-cover by stars of a graph
from $\mathcal G_d$ covers some vertex at least $(d+1)/2$ times, even the
ones which have large girth.  This is in sharp contrast to the easy fact
that a tree has a star-cover where each vertex is covered at most twice.

\subsection{Organization}

The {\em entropy method} is a universal, but not complete, tool to establish
lower bounds on the information ratio of arbitrary access structures, see
\cite{beimel,bssv6,bssv,csirm_lower,csirmcube,dijk}.  For the convenience of
the interested reader, the method is explained in Section \ref{sec:method}. 
Specific lemmas tailored for establishing lower bounds on graph-based access
structures are stated and proved in Section \ref{sec:general_lemmas}. 
Section \ref{const_sec} contains the definition of the family $\mathcal G_d$
of $d$-regular graphs along with the proof that they have complexity
$(d+1)/2$.  This graph family extends that of van Dijk from \cite{dijk} and
\cite{bssv}.  The existence of large girth graphs in $\mathcal G_d$ is
proved in Section \ref{graph_sec}.  The construction and the results in this
section are purely combinatorial, and may be interesting independently. 
Finally, Section \ref{sec:conclusion} concludes the paper and lists some
open problems.

\section{The entropy method}\label{sec:method}

We focus on the special case of graph based structures. The method can be
summarized as follows.  Consider any function $f$ assigning non-negative
real numbers to subsets of the vertices of $G$ (the normalized entropy
function) with the following properties:
\begin{enumerate}
\item $f$ is monotone and submodular; moreover $f(\emptyset)=0$;
\item if $A\subset B$, $A$ is independent and $B$ is not, then
$f(A)+1\le f(B)$ (strict monotonicity)
\item if $C$ is empty or independent, $AC$ and $BC$ are not independent
(qualified), then
$f(AC)+f(AB)\ge f(C)+f(ABC)+1$ (strict submodularity).
\end{enumerate}
If for any such function $f$ we have $f(v)\ge \alpha$ for some vertex
$v$ of $G$, then the complexity of $G$ is at least $\alpha$.

As in the formula above, throughout the paper we drop the $\cup$ sign when
denoting a union of subsets, and make no distinction between a vertex and
the one-element subset containing that vertex.  In particular, $aAB$ denotes
the subset $\{a\}\cup A\cup B$.  Lower case letters $a$, $b$, etc will
denote vertices of $G$, and capital letters $A$, $B$, etc denote subsets of
vertices.

The correctness of the method follows from the observation that any secret
sharing scheme is a collection of random variables, and $f(A)$ can be chosen
to be the total entropy of the shares given to the vertices in $A$ divided
by the entropy of the secret.  Thus the relative size of participant $v$'s
share is just $f(v)$.  The above properties of $f$ are the translations of
basic Shannon inequalities and the fact that the distributed secret is
independent from any independent set, and is determined by any
non-independent subset of $G$.  Consequently every secret sharing scheme
assigns a share of relative size $\alpha$ to some participant.

Following the entropy notation, we will use the following abbreviations:
\begin{align*}
\I_f(A;B) &\eqdef f(A)+f(B)-f(AB); \\
\I_f(A;B\|C) &\eqdef f(AC)+f(BC)-f(C)-f(ABC).\\
\end{align*}
The submodularity property gives that both expressions are non-negative,
moreover strict submodularity is equivalent to $\I_f(A;B\|C)\ge 1$ whenever
$C$ is either unqualified or empty, and $AC$ and $BC$ are qualified subsets.

In the formulas we frequently omit the function $f$, and any vertex, or
subset of vertices stand for its $f$ value as well.  In particular, we use
$\I(A;B)$ and $\I(A;B\|C)$ instead of $\I_f(A;B)$ and $\I_f (A;B\|C)$
whenever $f$ is clear from the context.

\section{Some general lemmas}\label{sec:general_lemmas}

We start by stating and proving some general lemmas which resemble to
information theoretic arguments.  Fix a function $f$ satisfying properties
1--3 defined in Section \ref{sec:method}.

\begin{lemma}\label{lemma:02}
If $C$ is independent, $AC$ and $BC$ are not, then $\I(A;B\|C)\ge 1$.
\end{lemma}
\begin{proof}
This is just the restatement of the strong submodularity requirement.
\end{proof}

Let $A$ and $B$ be disjoint subsets of the vertices. We say that there is a
{\em 1-factor from $B$ to $A$} if $B=\{b_1,\dots, b_t\}$, and there are $t$
different elements $a_1,\dots, a_t$ in $A$ such that $a_ib_j$ is an edge if
and only if $i=j$.

\begin{lemma}\label{lemma:3}
Let $A$ and $B$ be disjoint subsets of the vertices, $B$ is independent, $A$
is not independent such that there is a 1-factor from $B$ to $A$.  Then
$\I(A;B)\ge |B|$.
\end{lemma}
\begin{proof}
The proof goes by induction on the number of elements in $B$. When $B$ is
empty, then $|B|=0$ and there is nothing to prove.  Otherwise let $ab$ be an
edge in the 1-factor, and let the two sets be $A=aA^*$ and $B=bB^*$.  Then
$aB^*$ is independent, $abB^*$ and $aA^*B^*$ are not, thus we get
$\I(b;A^*\|aB^*)\ge 1$ by Lemma \ref{lemma:02}.  Then
\begin{align*}
\I(A;bB^*)-\I(A;B^*) &= \big(f(A)+f(bB^*) - f(bAB^*)\big) -
    \big(f(A)+f(B^*)-f(AB^*)\big) \\
    &= f(bB^*) - f(B^*) + f(aA^*B^*) - f(abA^*B^*) \\
    &= \I(b;A^*\|aB^*) + \I(a;b\|B^*) \ge 1 + 0.
\end{align*}
From here the induction hypothesis gives the claim.
\end{proof}

\begin{lemma}\label{lemma:4}
With the same assumptions as in Lemma \ref{lemma:3}, $f(A) \ge |B|+1$.
\end{lemma}
\begin{proof}
According to Lemma \ref{lemma:3}, $\I(A;B) = f(A)+f(B)-f(AB)
\ge |B|$. Thus
$$
    f(A) \ge |B| + \big(f(AB)-f(B)\big) \ge |B| + 1
$$
by strict monotonicity as $B$ is independent and $AB$ is not.
\end{proof}

\begin{lemma}\label{lemma:5}
Let $A$ and $B$ be disjoint subsets of the vertices such that neither $A$
nor $B$ is independent.  Suppose $B$ contains an independent subset $B'$ and
a 1-factor from $B'$ to $A$.  Then $\I(A;B)\ge |B'|+1$.
\end{lemma}
\begin{proof}
Let $B=B'B''$ where $B'$ is the independent set with the 1-factor. Then
$$
\I(A;B'B'') = \I(A;B') + \I(A;B''\|B') \ge |B'|+1
$$
by Lemmas \ref{lemma:3} and \ref{lemma:02}.
\end{proof}

\section{The graph family $\mathcal G_d$}\label{const_sec}

Let $n_2, n_3, \dots,$ be integers, each one is at least 5, and $n_2$ is
even.  We construct a sequence of bipartite graphs $G_2$, $G_3$, \dots as
follows.  $G_2$ is the cycle with $n_2$ nodes; $A_2$, $B_2$ are the two
independent sets consisting of the odd and even vertices, respectively.

Suppose $G_d$ has been constructed; the equal size partition $(A_d,B_d)$ of
its vertices shows that $G_d$ is bipartite; $A_d$ and $B_d$ are independent,
and all edges of $G_d$ go between $A_d$ and $B_d$.  (This property holds for
$G_d$ by induction.) Take $n_{d+1}$ copies of $G_d$ denoted as $G^i_d$ with
$G^{n_{d+1}+1}_d=G^1_d$.  The vertex set of $G^i_d$ is $A^i_d\cup B^i_d$
where $A^i_d$ and $B^i_d$ are the (equal size) independent subsets of the
vertices of $G^i_d$.
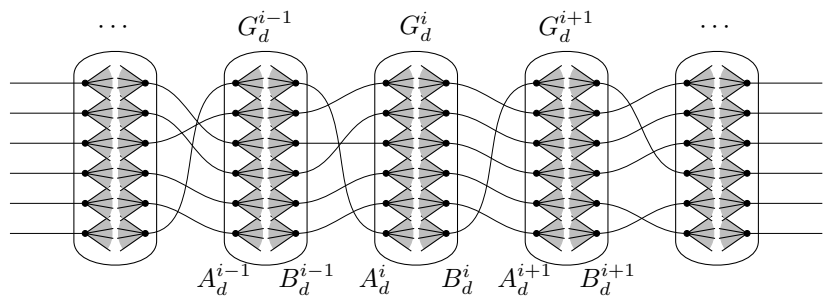
\begin{figure}[!ht]\begin{center}
\begin{tikzpicture}
\foreach \x in {0,...,4}{
\foreach \y in {1,...,6}{
\coordinate (a\x\y) at (\x*2.0-0.1, 0.4*\y);
\coordinate (b\x\y) at (0.7+\x*2.0,0.4*\y);
\filldraw[lightgray] (a\x\y) -- ++(-35:0.357) arc(-35:35:0.357) -- (a\x\y);
\filldraw[lightgray] (b\x\y) -- ++(145:0.357) arc(145:215:0.357) -- (b\x\y);
\draw[fill] (a\x\y) circle (1.1pt) (b\x\y) circle (1.1pt);
\draw (a\x\y) -- +(-35:0.41) (a\x\y) -- +(-10:0.38) (a\x\y)--+(10:0.38)
  (a\x\y) -- +(35:0.41);
\draw (b\x\y) -- +(145:0.41) (b\x\y)--+(170:0.38) (b\x\y)--+(190:0.38)
  (b\x\y) -- +(215:0.41);
}
\draw (\x*2.0-0.25,0.3)--(\x*2.0-0.25,2.5) to[out=90,in=90] (\x*2.0+0.85,2.5) --
(\x*2.0+0.85,0.3) to[out=-90,in=-90] (\x*2.0-0.25,0.3);
}
\foreach \y in{1,...,6}{
\draw[very thin] (a0\y) -- +(-1.0,0) (b4\y)-- +(1.0,0);
}
\foreach \x/\y in {1/6,2/1,3/2,4/5,5/3,6/4}
{ \draw[very thin] (b0\x) to[out=0,in=180] (a1\y); }
\foreach \x/\y in {1/2,2/3,3/5,4/4,5/6,6/1}
{ \draw[very thin] (b1\x) to[out=0,in=180] (a2\y); }
\foreach \x/\y in {1/6,2/1,3/2,4/3,5/4,6/5}
{ \draw[very thin] (b2\x) to[out=0,in=180] (a3\y); }
\foreach \x/\y in {1/2,2/1,3/4,4/5,5/6,6/3}
{ \draw[very thin] (b3\x) to[out=0,in=180] (a4\y); }
\draw (0.3,3.15) node {$\cdots$};
\draw (2.3,3.15) node {$G^{i-1}_d$};
\draw (4.3,3.15) node {$G^i_d$};
\draw (6.3,3.15) node {$G^{i+1}_d$};
\draw (8.3,3.15) node {$\cdots$};

\draw (1.75,-0.215) node {$A^{i-1}_d$};
\draw (2.85,-0.215) node {$B^{i-1}_d$};
\draw (3.75,-0.215) node {$A^i_d$};
\draw (4.85,-0.215) node {$B^i_d$};
\draw (5.75,-0.215) node {$A^{i+1}_d$};
\draw (6.85,-0.215) node {$B^{i+1}_d$};
\end{tikzpicture}

\kern -15pt\end{center}
\caption{Structure of the graph $G_{d+1}$}\label{fig:structure}
\end{figure}
To get $G_{d+1}$ add an (arbitrary) 1-factor between $B^i_d$ and $A^{i+1}_d$
for all $i=1,2,\dots, n_{d+1}$, see Figure \ref{fig:structure}.  The equal
size partition of $G_{d+1}$ showing that $G_{d+1}$ is bipartite is the union
of vertices in $A^i_d$ and the union of vertices in $B^i_d$, respectively. 
The graph family $\mathcal G_d$ consists of all graphs $G_d$ constructed
this way.

\begin{claim}
$G_d$ is a $d$-regular bipartite graph on $n_2n_3\cdots n_d$
vertices.
\end{claim}
\begin{proof}
Immediate from the definition.
\end{proof}

According to Stinson's theorem \cite{stinson}, the information ratio of a
$d$-regular graph is at most $(d+1)/2$.  The next theorem claims that this
bound is tight for graphs in $\mathcal G_d$.  As van Dijk's graph family is
contained properly in $\mathcal G_d$, this theorem extends the main result
of \cite{bssv}.

\begin{theorem}\label{thm:main}
For any normalized entropy function $f$ on $G_d\in \mathcal G_d$ the
following inequality holds:
\begin{equation}\label{eq:thm-main}
   \sum_{v\in G_d} f(v) \ge \frac{d+1}2\,|G_d| .
\end{equation}
Consequently the information ratio of $G_d$ is exactly $(d+1)/2$.
\end{theorem}

The crux of the proof is the following inequality which will be proved by
induction along the construction of the graph $G_d$.

\begin{lemma}\label{lemma:main}
For any normalized entropy function $f$ on $G_d$ the following inequality
holds:
\begin{equation}\label{eq:lemma-main}
   \sum_{v\in G_d} f(v) - f(G_d) \ge
   \frac{d}2\,|G_d| - 1.
\end{equation}
\end{lemma}

\noindent
Let us first see how to derive Theorem \ref{thm:main} from this lemma.

\begin{proof}[Proof of Theorem \ref{thm:main} assuming Lemma
\ref{lemma:main}]
The graph $G_{d+1}$ consists of $n_{d+1}$ copies of $G_d \in \mathcal G_d$;
each copy is connected by a 1-factor to the rest of $G_{d+1}$: half of the
edges go to the previous copy, the other half of the edges go to the next
copy; and the other endpoints of this 1-factor form an independent set of
size $|G_d|$, see Figure \ref{fig:structure}.  Thus, by Lemma \ref{lemma:4},
$$
    f(G^i_d) \ge |G_d| + 1.
$$
Applying Lemma \ref{lemma:main} to each copy $G^i_d$ separately, we get
\begin{align*}
  \sum_{v\in G_{d+1}} f(v)
       &= \sum_{i=1}^{n_{d+1}} \sum_{v\in G^i_d} f(v)
          \ge \sum_{i=1}^{n_{d+1}} \left(f(G^i_d)+ \frac{d}2\,|G_d| -1\right) \\
       &\ge \sum_{i=1}^{n_{d+1}} \left( |G_d| + |G_d|\frac{d}2\right)
        = \frac{d+2}2\,|G_{d+1}|,
\end{align*}
which proves (\ref{eq:thm-main}) when $d$ is at least 3. For the case $d=2$
let $a$, $b$, $c$, $d$ be four neighbor vertices on the cycle $G_2$. Now
$f(bc)\ge 3$ by Lemma \ref{lemma:4}, as $ad$ is independent, $bc$ is not,
and there is a 1-factor from $ad$ to $bc$. Consequently
$$
  \sum_{v\in G_2} f(v)
       \ge f(v_1v_2)+f(v_3v_4)+\cdots
       \ge \frac{|G_2|}2 \cdot 3 .
$$
Here we used $f(b)+f(c)\ge f(bc)$, which is a special case of submodularity.
\end{proof}

Now we turn to the proof of Lemma \ref{lemma:main}. The validity of the
following decomposition of an expression similar to the left hand side of
(\ref{eq:lemma-main}) follows immediately from the definitions.
\begin{fact}\label{fact:decomp}
Suppose $n\ge 5$ and $E_i$ are subsets of the vertices for $i=1,\dots,n$.
Then $\sum_{i} f(E_i) - f(E_1\dots E_n)$ can be written as the following 
sum of $3+(n-4)$ entropy terms:
\begin{align*}
  &\I(E_1;E_2) + \I(E_1E_2;E_3) + \I(E_1E_2E_3; E_4E_5\dots E_n) +{}\\
 {}+{}&\I(E_4;E_5) + \I(E_4E_5;E_6) + \cdots+\I(E_4E_5\dots E_{n-1};E_n).
\end{align*}
\end{fact}

\begin{proof}[Proof of Lemma \ref{lemma:main}]
First we prove (\ref{eq:lemma-main}) for the base case $d=2$. The graph
$G_2$ is the cycle on $n=|G_2|\ge 6$ vertices.  Let us denote the vertices
by $v_1,\dots, v_n$.  The edges of $G_2$ are $v_iv_{i+1}$ and $v_nv_1$.  For
this graph inequality (\ref{eq:lemma-main}) rewrites to
$$
    \sum_{i=1}^n f(v_i) - f(G_2) \ge n-1.
$$
To prove it we use the decomposition in Fact \ref{fact:decomp} with
$E_i=\{v_i\}$ to rewrite the left hand side as the $3+(n-4)$-term sum
\begin{align*}
  &\I(v_1;v_2) + \I(v_1v_2;v_3) + \I(v_1v_2v_3; v_4v_5\dots v_n) +{}\\
 {}+{}&\I(v_4;v_5) + \I(v_4v_5;v_6) + \cdots+\I(v_4v_5\dots v_{n-1};v_n).
\end{align*}
All the terms here are non-negative. Furthermore, by Lemma
\ref{lemma:3} we have
\begin{align*}
   \I(v_1v_2; v_3) &\ge 1, \\
   \I(v_4v_5; v_6) &\ge 1, \\
   \hdots & \\
   \I(v_4v_5\dots v_{n-1};v_n) &\ge 1;
\end{align*}
and by Lemma \ref{lemma:5}, $\I(v_1v_2v_3; v_4\dots v_n)\ge 3$ witnessed by
the independent set $v_1v_3$. These numbers add up to $n-1$ proving the
base case.

Next suppose we know the inequality (\ref{eq:lemma-main}) for the graph
$G_d$, and want to prove it for $G_{d+1}$.  The graph $G_{d+1}$ consists of
$n_{d+1}=n\ge 5$ copies of $G_d$.  The vertex set of the $i$-th copy is
$V_i=A_i\cup B_i$, where $A_i$ and $B_i$ are disjoint independent sets of
equal size.  The induction hypothesis tells us that
$$
    \sum_{v\in V_i} f(v) - f(V_i) \ge \frac{d}2\,|G_d|-1
$$
for every $i$. Since $|G_{d+1}| = n\cdot|G_d|$, we are done if we prove that
\begin{align}\label{eq:sufficient}
   \sum_{i=1}^n f(V_i) - f(G_{d+1})&\ge
    \left( \frac{d+1}2\,|G_{d+1}|-1 \right)
     - n\,\left(\frac{d}2\,|G_d|-1 \right) \nonumber\\
       &= \frac n2\,|G_d| +n-1 .
\end{align}
Apply the decomposition using $E_i=V_i$ to get
\begin{align*}
  &\I(V_1;V_2) + \I(V_1V_2;V_3) + \I(V_1V_2V_3; V_4\dots V_n) +{}\\
 {}+{}&\I(V_4;V_5) + \I(V_4V_5;V_6) + \cdots+\I(V_4\dots V_{n-1};V_n)
\end{align*}
as an equivalent form of the left hand side here.  By Lemma \ref{lemma:5}
each of these quantities, except for the third one, have value at least
$1+|G_d|/2$, as $A_i\subset V_i$ is an independent set of size $|G_d|/2$
connected to the other part by a 1-factor.  As $V_1V_2V_3$ has an
independent set of size $|G_d|$ (the vertex set $A_1B_3$) connected by a
1-factor to $V_4\dots V_n$, we have
$$
    \I(V_1V_2V_3; V_4\dots V_n) \ge |G_d| + 1.
$$
The sum of these values is $n|G_d|/2 + n-1$, as was required.
\end{proof}

\section{Graphs in $\mathcal G_d$ with large girth}\label{graph_sec}

Graphs in $\mathcal G_2$ are even length cycles, thus the girth is equal to
the number of vertices -- which can be arbitrary large.  The first challenge
is to find large girth graphs in $\mathcal G_3$ where one can choose the
1-factors between the neighboring independent sets arbitrarily, see Figure
\ref{fig:structure}.  However, it is not clear how to control the
interaction of those choices in order to avoid introducing short cycles.  A
natural approach would be choosing these 1-factors randomly.  With too many
1-factors the graph will have constant girth with overwhelming probability,
but maybe Lovasz Local Lemma \cite{lll} can be used to show that the girth
is $>g$ with non-zero (while exponentially small) probability.
Unfortunately this approach failed as well the attempts to use algebraic
construction.

Our method which finds large girth graphs in $\mathcal G_{d+1}$ is based on
the following idea.  We choose all 1-factors between the different copies of
the $G_d$ graph identically.  Then map $G_{d+1}$ to $G_d$ so that the image
of the vertex $v^i$ in the $i$-th copy of $G_d$ is just $v$.  Then the image
of $G_{d+1}$ is $G^*_d$ which has the same edges as $G_d$ plus a (maximal)
1-factor between the two independent sets $A_d$ and $B_d$ of $G_d$.  {\em
If} we know that $G^*_d$ has no short cycles, {\em then} we can conclude
that neither does $G_{d+1}$.  Lemma \ref{lemma:c1} shows the existence of
such a $G^*_2$ graph on every large enough vertex number.  To ease the
description first we introduce the notion of $\pi$-graphs.

Let $\pi$ be a permutation on $1,\dots,n$.  A {\em $\pi$-graph} is a
3-regular bipartite graph on vertices $a_i$, $b_i$ for $1\le i\le n$ such
that $a_i$ is connected to $b_i$, to $b_{i+1}$ and to $b_{\pi(i)}$, where
$b_{n+1}=b_1$.

A $\pi$-graph is bipartite, the two independent vertex classes are $\{a_i\}$
and $\{b_i\}$.  The edges $a_i$\cto$b_{\pi(i)}$ form a (maximal) 1-factor.

\begin{lemma}\label{lemma:c1}
There is a function $N(g)$ such that for $g>3$ and $n>N(g)$ there is a
permutation $\pi$ and a $\pi$-graph on $2n$ vertices which has girth $> g$.
\end{lemma}

\begin{proof}
The {\em distance} of two vertices is the length of the shortest path
between them.  Thus $G$ has girth ${}>g$ if for any edge $uv$ of $G$, after
deleting that edge the distance of $u$ and $v$ is at least $g$.

We construct the claimed $\pi$-graph by adding a 1-factor to a large cycle
on $2n$ vertices.  The initial graph has vertices $a_i$, $b_i$ for
$i=1,\dots,n$ and edges $a_i$\cto$b_i$ and $a_i$\cto$b_{i+1}$, where indices
are understood modulo $n$.  A new edge can be added between the 2-degree
vertices $a_i$ and $b_j$ if their distance is at least $g$.  For a 2-degree
vertex $v$ the number of vertices with distance $<g$ from $v$ is at most
$$
    1+2+4+\cdots+2^{g-1}<2^g .
$$
Consequently one can add the next edge in the 1-factor in a greedy way until
the number of (unmatched) 2-degree vertices among the $a_i$ vertices goes
below $2^g$.

At this point the graph has exactly $(2^g-1)+(2^g-1)$ 2-degree vertices, and
girth ${}>g$. Next find (circular) intervals $I_s$ of the indices
$1,\dots,n$ for $1\le s < 2^g$ with the following properties:
\begin{enumerate}
\item $I_s$ contains $2^g+1$ indices;
\item if $i\in I_s$, then both $a_i$ and $b_i$ are at distance $\ge g$ from
any 2-degree vertex;
\item if $i\in I_s$ and $j\in I_t$, $s\not= t$, then $a_i$ and $b_j$ has
distance $\ge g$.
\end{enumerate}
After picking $s$ intervals, the number of indices which cannot be the
midpoint of the next interval $I_{s+1}$ is at most
$$
    2(2^g+s\cdot(2^g+1))\,2^g\,(2^g+1).
$$
Consequently, if $n>2^{4g+3}$ then one can find such intervals. From each
interval $I_s$ pick two indices $u_s < \ell_s$ such that the distance
between the vertices $b_{u_s}$ and $a_{\ell_s}$ is at least $g$. As the
interval has length $2^g+1$, such a pair always exists. By construction,
all distances between the vertices in the set
$$
   \{a_{\ell_s},b_{u_s} \,:\, 1\le s < 2^g\}
$$
are at least $g$, moreover any of them is at distance $\ge g$ from any
2-degree vertex.

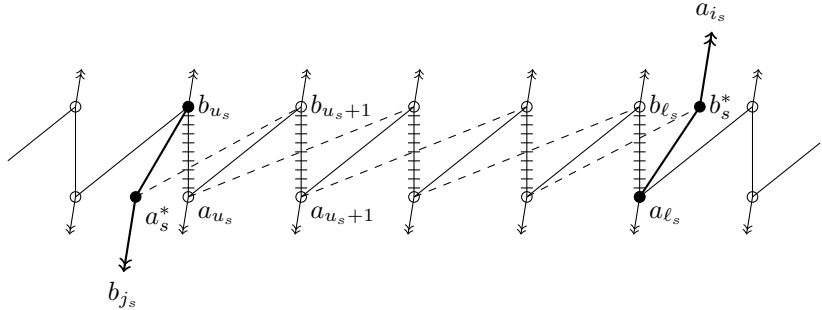
\begin{figure}[!ht]\begin{center}
\begin{tikzpicture}
\foreach \x in {1,...,7}{
\coordinate (a\x) at (1.5*\x,1.2);
\coordinate (b\x) at (1.5*\x,0);
}
\coordinate (astar) at (1.5*6+0.8,1.2);
\coordinate (bstar) at (1.5+0.8,0);
\draw (bstar) node [below right] {$a^*_s$}
      (astar) node [right] {$b^*_s$};
\draw (b2) node [below right] {$a_{u_s}$}
      (b3) node [below right] {$a_{u_s+1}$}
      (a2) node [right] {$b_{u_s}$}
      (a3) node [right] {$b_{u_s+1}$}
      (b6) node [below right] {$a_{\ell_s}$}
      (a6) node [right] {$b_{\ell_s}$};
\foreach \x in {1,...,7}{
\draw (a\x) circle (2pt);
\draw[->>] (a\x) -- +(0.08,0.5);
\draw (b\x) circle (2pt);
\draw[->>] (b\x) -- +(-0.08,-0.5);
}
\draw[->>,thick] (bstar) -- +(-0.16,-1.0) node [below] {$b_{j_s}$};
\draw[->>,thick] (astar) -- +(0.16,1.0) node [above] {$a_{i_s}$};
\draw[fill] (bstar) circle(2pt) (astar) circle(2pt)
    (a2) circle(2pt) (b6) circle(2pt);
\draw[thin, dashed] (bstar)--(a3) (b2)--(a4)
  (b3)--(a5) (b4)--(a6) (b5)--(astar);
\draw (a1)--(b1)--(a2) (b2)--(a3) (b3)--(a4) (b4)--(a5)
    (b5)--(a6) (b6)--(a7)--(b7)--+(0.6*1.5,0.6*1.2);
\draw (a1)--+(-0.6*1.5,-0.6*1.2);
\draw[thick] (bstar)--(a2) (astar)--(b6);
\foreach \x in {2,...,6}{
\draw (a\x)--(b\x);
\foreach \y in {1,...,9}{
    \draw[thin] (a\x) ++(-0.08,-\y*0.12) -- +(+0.16,0);
}
}
\end{tikzpicture}

\kern -20pt\end{center}
\caption{Swapping edges between $a^*_s$ and $b^*_s$ in $G$}\label{fig:graph1}
\end{figure}
Let the 2-degree vertices be $a_{i_s}$ and $b_{j_s}$ where $s$ runs from $1$
to $2^g-1$.  Add $2(2^g-1)$ new vertices and some new edges to the graph. 
The new vertices are $a^*_s$ and $b^*_s$, and the new edges are
$b_{j_s}$\cto$a^*_s$, $a^*_s$\cto$b_{u_s}$, and $a_{i_s}$\cto$b^*_s$,
$b^*_s$\cto$a_{\ell_s}$, see Figure \ref{fig:graph1}.  In this new graph $G$
the new vertices $a^*_s$ and $b^*_s$ have degree 2; vertices $b_{u_s}$ and
$a_{\ell_s}$ have degree 4, all other vertices have degree 3.  $G$ still has
girth $>g$.  Indeed, any cycle without the new vertices has length $>g$.  As
new vertices have degree 2, if a cycle contains $a^*_s$, then it also
contains $b_{j_s}$ and $b_{u_s}$; similarly for $b^*_s$.  A short cycle
cannot contain a single new vertex $a^*_s$ (or $b^*_s$) only, as $b_{j_s}$
and $b_{u_s}$ are far from each other.  So the short cycle contains (at
least) two new vertices connected by a short path, and this short path can
only be $a^*_s$\cto$b_{j_s}$--$\,\cdots$--\,$a_{i_t}$\cto$b^*_t$, or
$a^*_s$\cto$b_{j_s}$--$\,\cdots$--$b_{j_t}$\cto$a^*_t$, or
$b^*_s$\cto$a_{i_s}$--$\,\cdots$--$a_{i_t}$\cto$b^*_t$.  Any two of the
endpoints of these paths are far from each other, thus one cannot merge them
into a short cycle.

Now $G$ is bipartite, has girth $>g$, but it is not 3-regular. To make it
3-regular, we make the following changes.  Delete the edges between $a_i$
and $b_i$ for $u_s\le i\le \ell_s$ (after this $a_{\ell_s}$ and $b_{u_s}$
will have degree 3), and add the edges $a^*_s$\cto$b_{u_s+1}$,
$a_{\ell_s-1}$\cto$b^*_s$, and the edges $a_{i-1}$\cto$b_{i+1}$ for
$u_s<i<\ell_s$ as depicted on Figure \ref{fig:graph1}.  Denote this new
graph by $G^*$.

If $G^*$ has a cycle, then the newly added (dashed) edges in the cycle can
be replaced by three edges of $G$.  This way the cycle cannot collapse, and
after deleting edges traversed back and forth, the cycle in $G$ has at least
as many edges as it had in $G^*$, but not more than three times as much. 
Consequently, the girth of $G^*$ is bigger than $g/3$.

Summing up, for any $n>2^{4g+3}$ we have constructed a $\pi$-graph on
$2n+2(2^g-1)$ vertices which has girth bigger than $g/3$.  Thus we have
proved the lemma with the function $N(g)=2^{12g+4}$.  \end{proof}

Before continuing, let us show how Lemma \ref{lemma:c1} can be used to find
a large girth graph in the family $\mathcal G_3$. Fix a $\pi$-graph $G$ on
$2n$ vertices which has girth $>g$ as given by Lemma \ref{lemma:c1}. Let the
vertices of $G$ be labeled as $a_i$, $b_i$ such that $b_1$, $a_1$, $b_2$,
$a_2$, \dots, $b_n$, $a_n$ is a cycle, and the additional 1-factor is given
by the edges $a_i$\cto$b_{\pi(i)}$.

Let $m\ge 5$. The vertices of $H$ are $a_{i,j}$ and $b_{i,j}$ where $1\le
i\le n$ and $1\le j\le m$, the indices are understood modulo $n$ and $m$,
respectively.  $a_{i,j}$ is connected to $b_{i,j}$, $b_{i+1,j}$ and
$b_{\pi(i),j+1}$.  For fixed $j$, the vertices $a_{i,j}$ and $b_{i,j}$ form
a cycle on $2n$ vertices -- thus an instance of $\mathcal G_2$; and the
edges $a_{i,j}$\cto$b_{\pi(i),j+1}$ form a 1-factor between the independent
sets $A_j$ and $B_{j+1}$.  Consequently $H$ is a $\mathcal G_3$-graph.  The
map $\phi(a_{i,j})=a_i$, $\phi(b_{i,j})=b_j$ is a graph homomorphism from
$H$ to $G$, which maps any cycle in $H$ into a cycle in $G$.  (As $\phi$ is
onto, and both $H$ and $G$ are 3-regular, in the image no edge is traversed
immediately backward.) As $G$ has girth $>g$, $H$ has girth $>g$ as well, as
required.

\smallskip

To formalize the above construction, we introduce the following notation. 
Let $G$ be a graph on even number of vertices, $A$ and $B$ be the equal size
partition of the vertices, and $\pi: A\to B$ be a one-to-one mapping.  For
any $m>1$ the graph $\mathcal H(m,G,\pi)$ consists of $m$ disjoint copies of
$G$ labeled as $G^1$, $G^2$, \dots, $G^m$ with $G^{m+1}=G^1$.  There is an
additional 1-factor between $A_j$ of $G^j$ and $B_{j+1}$ of $G^{j+1}$
determined by the mapping $\pi$: the vertex $a_j\in A_j$ is connected to
$\pi(a)_{j+1}\in B_{j+1}$.  The above graph $H$ is just $\mathcal
H(m,C_{2n},\pi)$, where $\pi$ is the map from Lemma \ref{lemma:c1}.

\smallskip

For handling larger degree $d$ the following stronger claim will be proved
by induction of $d$.

\begin{lemma}\label{lemma:graph-induction}
For each $g>3$ there exists a $d$-regular graph $G_d\in \mathcal G_d$ with
independent vertex sets $A_d$, $B_d$, and a one-to-one map $\pi_d: A_d \to
B_d$ such that for all $m>1$ the graph $\mathcal H(m,G_d,\pi_d)$ has girth
$>g$.
\end{lemma}
\begin{proof}
For $d=2$ the graph $G_2$ is the cycle on $N_2=2n\approx 2\cdot2^{12g+4}$
vertices, and $\pi_2$ is the map given by Lemma \ref{lemma:c1}.  It
satisfies the claim of the lemma as discussed above.

Suppose we have the graph $G_d$ on even number of vertices and the map
$\pi_d$ as in the Lemma, and want to find $G_{d+1}$ and $\pi_{d+1}$.  Let
$t=3|G_d|$, and apply Lemma \ref{lemma:c1} to the girth $gt$ to get an
$N_{d+1}\ge 5$ which is a multiple of $|G_d|$, and the map $\pi$ on
$1,\dots,N_{d+1}$, which is a one-to-one map between the odd and even
indices.  Let $m_d=N_{d+1}/|G_d|$.  This is $\ge 5$, thus $G_{d+1}=\mathcal
H(m_d,G_d,\pi_d)$ is in $\mathcal G_{d+1}$, and has $N_{d+1}$ vertices. 
These vertices can be labeled as $v_i$ for $1\le i \le N_{d+1}$ such that
\begin{enumerate}
\item[a)] the two independent sets of $G_{d+1}$ are the odd-indexed 
vertices and the even-indexed vertices;
\item[b)] if $v_i$\cto$v_j$ is an edge
in $G_{d+1}$, then the circular distance between $i$ and $j$ (that is, the 
smaller of $|i-j|$ and $N_{d+1}-|i-j|)$) is at most $t$.
\end{enumerate}
Indeed, as $t=3|G_d|$, enumerate the vertices in the first copy of $G_d$
first, then the vertices in the second copy, and so on, making sure that
the independent sets get the even and odd indices, respectively. Finally let
$\pi_{d+1}$ be the map between the odd and even numbers between $1$ and
$N_{d+1}$ as induced by the map $\pi$ of Lemma \ref{lemma:c1}.

We claim that $G_{d+1}$ and $\pi_{d+1}$ satisfies the conditions of the
lemma.  The only non-trivial case is that $H=\mathcal
H(m,G_{d+1},\pi_{d+1})$ has girth $>g$ for every $m>1$.  Consider a (short)
cycle in $H$.  If this cycle has vertices from the same copy of $G_{d+1}$
only, then it has length $>g$ as $G_{d+1}$ has girth $>g$ by induction.

Consequently the cycle must have vertices in different copies of $G_{d+1}$.
Let $v^j_{i_1}$ and $v^j_{i_2}$ be neighbor points in the cycle in the
$j$-th copy, where $1\le i_1,i_2\le N_{d+1}$.  By construction, the circular
distance of $i_1$ and $i_2$ is at most $t$, thus the
$v^j_{i_1}$\cto$v^j_{i_2}$ edge can be replaced by at most $t$ edges now in
the cycle $v_1,v_2,\dots,v_{N_{d+1}-1},v_{N_{d+1}}$.  But it means that the
graph returned by Lemma \ref{lemma:c1} has a cycle of length $\le gt$, which
is a contradiction.
\end{proof}

The main theorem of this section is a simple corollary of Lemma
\ref{lemma:graph-induction}.

\begin{theorem}\label{thm:last}
There are arbitrary large girth graphs in $\mathcal G_d$.
\end{theorem}
\begin{proof}
As $H=\mathcal H(m,G_d,\pi_d)$ has girth $>g$, and $G_d$ is a spanned subgraph
of $H$, $G_d$ must have girth $>g$ as well.
\end{proof}

\section{Conclusion}\label{sec:conclusion}

Using the entropy method, it was shown that the general upper bound
$(d+1)/2$ on the complexity of graph based secret sharing schemes, known as
Stinson's bound, is tight for a large class of inductively defined
$d$-regular bipartite graphs.  The class $\mathcal G_d$ contains the graphs
defined by van Dijk \cite{dijk} and proved to be tight by Blundo et al
\cite{bssv}.

In Section \ref{graph_sec} it was proved that each $\mathcal G_d$ contains
graphs of arbitrary large girth using combinatorial arguments.  This result
refutes the widely believed conjecture that large girth graphs have bounded
complexity -- due to the exponentially diminishing interaction between the
shares assigned to the vertices.

\smallskip

The construction of the graph family $\mathcal G_d$ in Section
\ref{const_sec} is not the most general one. When connecting instances of
$G_d$ by 1-factors as indicated on Figure \ref{fig:structure}, the instances
$G^i_d$ need not be isomorphic. The construction uses that they are of equal
size, and the proof uses only that all of them are from $\mathcal G_d$.
Stinson's bound is tight for graphs in this extended family of bipartite
$d$-regular graphs.

While graphs in $\mathcal G_d$ can have arbitrary large girth, they are
highly connected: between any two vertices there are $d$ edge-disjoint
paths. High connectivity, however, is compatible with low complexity: the
complete graph $K_n$ has complexity 1.

\smallskip

The proof of Theorem \ref{thm:last} stating that $\mathcal G_d$ contains
graphs of girth $>g$, can be used to estimate the {\em size} of this graph.
$N_2 \approx g$ as $\mathcal G_2$ contains cycles; $N_3 \approx
12\cdot2^{12g+4}$ by Lemma \ref{lemma:c1}; and the inductive construction of
Lemma \ref{lemma:graph-induction} implies $N_{d+1} \approx
12\cdot2^{36\,g\,N_d}$, a really huge number. A $d$-regular graph with girth
$>g$ must contain at least $d^g$ vertices, and there are such graphs with
essentially that many vertices. It seems plausible that the family $\mathcal
G_d$ contains a $>g$ girth graph with about that much vertices. We leave to
prove (or refute) this claim as an open problem.

\section*{Acknowledgment}

This research was partially supported by the Lend\"ulet program of the
Hungarian Academy of Sciences.  The authors thank the members of the Crypto
Group of the R\'enyi Institute, and especially G\'abor Tardos, the numerous
insightful discussions on the topic of this paper.

\end{document}